\PassOptionsToPackage{usenames,dvipsnames}{xcolor}
\documentclass[conference]{IEEEtran}
\usepackage{hyperref}
  \hypersetup{colorlinks=true,citecolor=blue,linkcolor=blue,urlcolor=blue}
\usepackage{graphicx,amsthm,soul,amsmath,amssymb,subfig}
\graphicspath{ {./images/} }
\usepackage{booktabs}
\usepackage{float}

\theoremstyle{plain}

  \theoremstyle{plain}
  
  \theoremstyle{definition}
  
  \theoremstyle{plain}
  \newtheorem{cor}{Corollary}
  \theoremstyle{plain}
\newtheorem{prop}{Proposition}
\theoremstyle{plain}

\ifCLASSOPTIONcompsoc
  \usepackage{cite}
\else
  \usepackage{cite}
\fi
\ifCLASSINFOpdf
\else
\fi
\usepackage{tabularx}
\usepackage{array}
\begin{document}
%
\title{Fragmentation of Distributed Exchanges}

\author{\IEEEauthorblockN{Marius Zoican}
\IEEEauthorblockA{University of Toronto Mississauga \\
Rotman School of Management \\
Toronto, ON, Canada \\
Email: \url{marius.zoican@utoronto.ca}}
\and
\IEEEauthorblockN{Sorin Zoican}
\IEEEauthorblockA{Politehnica University of Bucharest \\
Department of Telecommunications \\
Bucharest, Romania \\
Email: \url{sorin@elcom.pub.ro}}
}


%


\maketitle

\begin{abstract}
Distributed securities exchanges may become \emph{de facto} fragmented if they span geographical regions with asymmetric computer infrastructure. First, we build an economic model of a decentralized exchange with two miner clusters, standing in for compact areas of economic activity (e.g., cities). ``Local'' miners in the area with relatively higher trading activity only join a decentralized exchange if they enjoy a large speed advantage over ``long-distance'' competitors. This is due to a transfer of economic value across miners, specifically from high- to low-activity clusters. Second, we estimate the speed advantage of ``local'' over ``long-distance'' miners in a series of Monte Carlo experiments over a two-cluster, unstructured peer-to-peer network simulated in \texttt{C}. We find that the speed advantage increases in the level of infrastructure asymmetry between clusters. Cross-region DEX blockchains are feasible as long as the asymmetry levels in trading activity and infrastructure availability across regions are positively correlated.
\end{abstract}

\begin{IEEEkeywords}
distributed exchange, fragmentation, P2P network, Monte Carlo simulation, financial markets
\end{IEEEkeywords}


%
\IEEEpeerreviewmaketitle

\section{Introduction}
In today's fragmented equity markets, the same security can be traded simultaneously on multiple exchanges. For example, in the United States, no fewer than 13 trading venues compete to offer trading services \cite{ComertonForde2019}.

At a first glance, trading fragmentation is socially sub-optimal \cite{Pagano1989}. Exchanges are natural monopolies: traders naturally gravitate towards ``thick'' marketplaces to improve their chances of finding a counterparty. On the one hand, buyers and sellers are more likely to be matched if they submit orders to the same exchange. On the other hand, a monopolist exchange can charge high fees and commissions, potentially distorting incentives to trade. 

The emergence of electronic trading opened the opportunity for low-cost market fragmentation. Trading algorithms can (quasi-) simultaneously monitor available prices on multiple exchanges and automatically choose the best one within milliseconds. As a result, technology eroded the natural monopoly position of exchanges. Spurred on by the shift in the technology frontier, two regulatory reforms, RegNMS in the United States (2005) and Mifid in Europe (2007) directly mandated competition between platforms, setting the stage for today's fragmented trading environment.  

In theory, \emph{distributed exchanges} (DEX) as proposed in \cite{AuroraLabs2017IDEXExchange} or \cite{Warren20170x:Blockchain} can achieve the best of both worlds: a centralized marketplace where the infrastructure provider lacks monopoly power. In a DEX, computers supply trading infrastructure competitively by ``mining'' orders on a distributed ledger, e.g., a Blockchain. We refer to such computers as miners throughout the rest of the paper. In contrast to a traditional exchange which can unilaterally set trading fees, miners strategically interact with each other: That is, charging a lower fee allows a miner to steer order flow away from competitors. At the same time, if miners post orders to a unique order book within a distributed smart contract, a centralized market-place emerges where all buyers and all sellers participate to the same network. 

This papers' contribution is to examine whether such a centralized, competitive marketplace can emerge in a distributed exchange setting. Section \ref{sec:econ} builds a simple economic model. We show that if trading interest is unevenly distributed across regions, a single DEX implies a transfer of economic value from high-volume to low-volume regions. A single DEX emerges only if local miners have a large enough processing speed advantage over long-distance miners, relative to any economies of scope from managing a unique blockchain rather than multiple fragmented ones. Further, miners' incentives to become faster act as a centripetal force generating fragmentation of distributed exchanges: If miners use more computing power to be locally faster, the network latency between regions becomes relatively more important. 

Having established that the ``local'' miners' speed advantage is a key factor for their incentives to set up a DEX,  we calibrate the parameter in Section \ref{sec:CPU}. To this end, we set up Monte Carlo simulation experiments over an unstructured, two cluster, peer-to-peer network. Further, we analize the sensitivity of local miners' speed advantage to factors such as the distance between regions and the size asymmetry between clusters. We find that local miners' advantage is higher if regions are farther apart geographically and if computer clusters are of very different sizes. 

Several papers study distributed exchanges, and blockchain financial markets, from an academic perspective. Closest to our setting, \cite{BrolleyZoican2019} argue that introducing a DEX could reduce the cost of idle trading infrastructure by generating a real-time price for trading speed. Using blockchain data, \cite{FlashBoys20} document that traders on decentralized exchanges engage in bidding wars on fees to obtain time priority. To prevent front-running on decentralized exchanges, \cite{Aune2017} proposes an encryption algorithm for orders themselves. \cite{Easley2019} build a game-theoretical model to examine the strategic interaction between miners and traders. 

There is an extensive literature in financial economics on the impact of market fragmentation. Studying U.S. markets, \cite{OHaraYe2011} find that fragmentation is associated with higher price volatility and greater market efficiency. The existence of multiple time-priority order queues enables queue-jumping by impatient traders, leading to an increase in overall market depth \cite{FoucaultMenkveld2008}. On the other hand, fragmentation can lead to market instability, especially in an environment dominated by algorithmic trading  \cite{MenkveldYueshen2019}.

\section{A simple economic framework \label{sec:econ}}
We build a simple economic framework to model the miners' problem as a function of messaging speed over the Blockchain network.

Consider a two-period economy where miners and traders are clustered in two geographical areas, \textbf{A} (large) and \textbf{B} (small). We interpret the areas as two large cities or metropolitan areas (e.g., akin to trading between New York an Chicago).  Each area hosts $N\geq 2$ miners, such that the trading infrastructure is symmetric. There is a unit mass of traders across clusters, not necessarily evenly distributed. Without loss of generality, a fraction $\beta\in\left[\frac{1}{2},1\right]$ of traders are located in area \textbf{A} and $1-\beta$ in area \textbf{B} respectively. 

We assume that miners receive an exogenous fee $f\geq 0$ for each order they are the first to process. Further, miners are able to process ``local'' transactions (i.e., same-city trades) faster than ``long-distance'' transaction. Particularly, processing times for local and long-distance transactions are exponentially distributed with rate $\lambda$ and $\pi\lambda$, respectively, where $\pi\leq 1$. From the properties of a Poisson process, it follows that the probability of a miner being first to process a local transaction (competing against both local and long-distance miners) is 
\begin{equation}
    \mathbb{P}\left(\text{win local trade}\right)=\frac{\lambda}{N\lambda+N\pi\lambda}=\frac{1}{N\left(1+\pi\right)}.
\end{equation}
Conversely, the probability of a miner to first process a long-distance transaction fee is
\begin{equation}
    \mathbb{P}\left(\text{win long-distance trade}\right)=\frac{\pi}{N\left(1+\pi\right)}.
\end{equation}

\subsection{Miner profits and local speed advantage}

The expected profit of a miner in area \textbf{A} is
\begin{equation}\label{eq:piL}
    \mathbb{E}\text{Profit}_\textbf{A}=\left[\beta \frac{1}{N\left(1+\pi\right)} + \left(1-\beta\right) \frac{\pi}{N\left(1+\pi\right)}\right] f.
\end{equation}
At the same time, the expected profit of a miner in area \textbf{B} is
\begin{equation}\label{eq:piS}
    \mathbb{E}\text{Profit}_\textbf{B}=\left[\beta \frac{\pi}{N\left(1+\pi\right)} + \left(1-\beta\right) \frac{1}{N\left(1+\pi\right)}\right] f.
\end{equation}

\begin{prop}
In a single-blockchain setup, miners in the large region \textbf{A} outearn miners in the small region \textbf{B} if 
\begin{enumerate}
    \item[(a)] traders tend to cluster in region \textbf{A} (if  $\beta>\frac{1}{2}$), and
    \item[(b)] local mining is more efficient than long-distance mining, that is $\pi<1$.
\end{enumerate}
\end{prop}
\begin{proof}
We compute the difference between \eqref{eq:piL} and \eqref{eq:piS},
\begin{equation}
    \mathbb{E}\text{Profit}_\textbf{A}-\mathbb{E}\text{Profit}_\textbf{B}=\frac{\left(1-\pi\right)\left(2\beta-1\right)}{N\left(1+\pi\right)}f \geq 0,
\end{equation}
as long as $\beta\geq\frac{1}{2}$ and $\pi\leq 1$.
\end{proof}

It turns out that profit asymmetry between miners in different regions is driven by both ex ante skewness in traders' geographical distribution \emph{and} by the speed advantage of local over long-distance miners. 

We benchmark the outcome against a fragmented blockchain setup. With fragmentation, miners have exclusive access to local traders. Therefore, profits in each area are proportional to the mass of traders, $\beta$ and $1-\beta$:
\begin{align}\label{eq:fragmented}
    \mathbb{E}\text{Profit}^\text{fragmented}_\textbf{A}&=\frac{\beta}{N}f \text{ and} \\
    \mathbb{E}\text{Profit}^\text{fragmented}_\textbf{B}&=\frac{1-\beta}{N}f.
\end{align}

Comparing equations \eqref{eq:piL}, \eqref{eq:piS}, and \eqref{eq:fragmented}, we conclude that a single blockchain always benefits miners in small areas and harms miners in large areas. In other words, there is a symmetric value transfer from regions with high trading volume to regions with little trading activity.

The key variable driving the magnitude of the transfer is the speed ratio between long-distance and local miners. If local miners have a large comparative advantage ($\pi \searrow 0$), markets are de facto fragmented with each set of miners processing local transactions only. Conversely, if local miners have little comparative advantage ($\pi \nearrow 1$), markets are consolidated and miners in both areas earn the same expected profit.

An important corollary is that a single blockchain dampens the effects of any initial asymmetry in regions' trading interest, inasmuch as miners have a positive probability of capturing long-distance trades.

\subsection{Endogenous adoption}

We introduce a cost for maintaining a blockchain, $\mathcal{C}\left(x\right)\geq 0$ -- where $x$ is the DEX trading volume. The cost function has the property that $\mathcal{C}\left(a x\right)< a \mathcal{C}\left(x\right)$: That is, it is more expensive to maintain two small DEX blockchains than a single DEX with double the transaction volume. Moving from a fragmented DEX to a single DEX generates cost savings for each miner, since
\begin{equation}\label{eq:cost}
    \text{Cost savings ($\Delta$)} = \frac{\mathcal{C}\left(x\right)}{N}-\frac{\mathcal{C}\left(2x\right)}{2N}\geq 0.
\end{equation}

Since moving from an unique to a fragmented DEX implies a zero-sum transfer between miner groups, a single DEX is socially optimal as it entails half the fixed cost of a fragmented exchange. 

Will a single blockchain be adopted in equilibrium? Our results suggest that miners in low-volume areas always prefer a unique blockchain to a fragmented market: both due to the positive rent transfer, and the cost savings. For miners in large-volume areas, a single blockchain is optimal if and only if the cost savings are large enough, or
\begin{align}
    \mathbb{E}\text{Profit}^\text{fragmented}_\textbf{A}-\mathbb{E}\text{Profit}_\textbf{A}& <\Delta \Longrightarrow \\
    \frac{\pi\left(2\beta-1\right)}{N\left(1+\pi\right)}& <\Delta.
\end{align}

\begin{cor}
\textbf{A}-area miners only adopt a single blockchain if they have a strong local-to-long-distance speed advantage, that is for $\pi\leq \bar{\pi}$ where
\begin{equation}
    \bar{\pi}\equiv \frac{N\Delta}{\left(2\beta-1\right)-N\Delta}.
\end{equation}
\end{cor}

Figure \ref{fig:theory} illustrates the miners' profits as a function of relative speed, in both single-blockchain and fragmented-blockchain market setups.

\begin{figure}[h]
\caption{\label{fig:theory} \textbf{Miner profits and market structure}}
\begin{minipage}[t]{1\columnwidth}%
\footnotesize
Parameter values: $\beta=0.75$, $N=4$, and $f=1$.
\end{minipage}
\vspace{0.1in}
\begin{centering}
\includegraphics[width=\columnwidth]{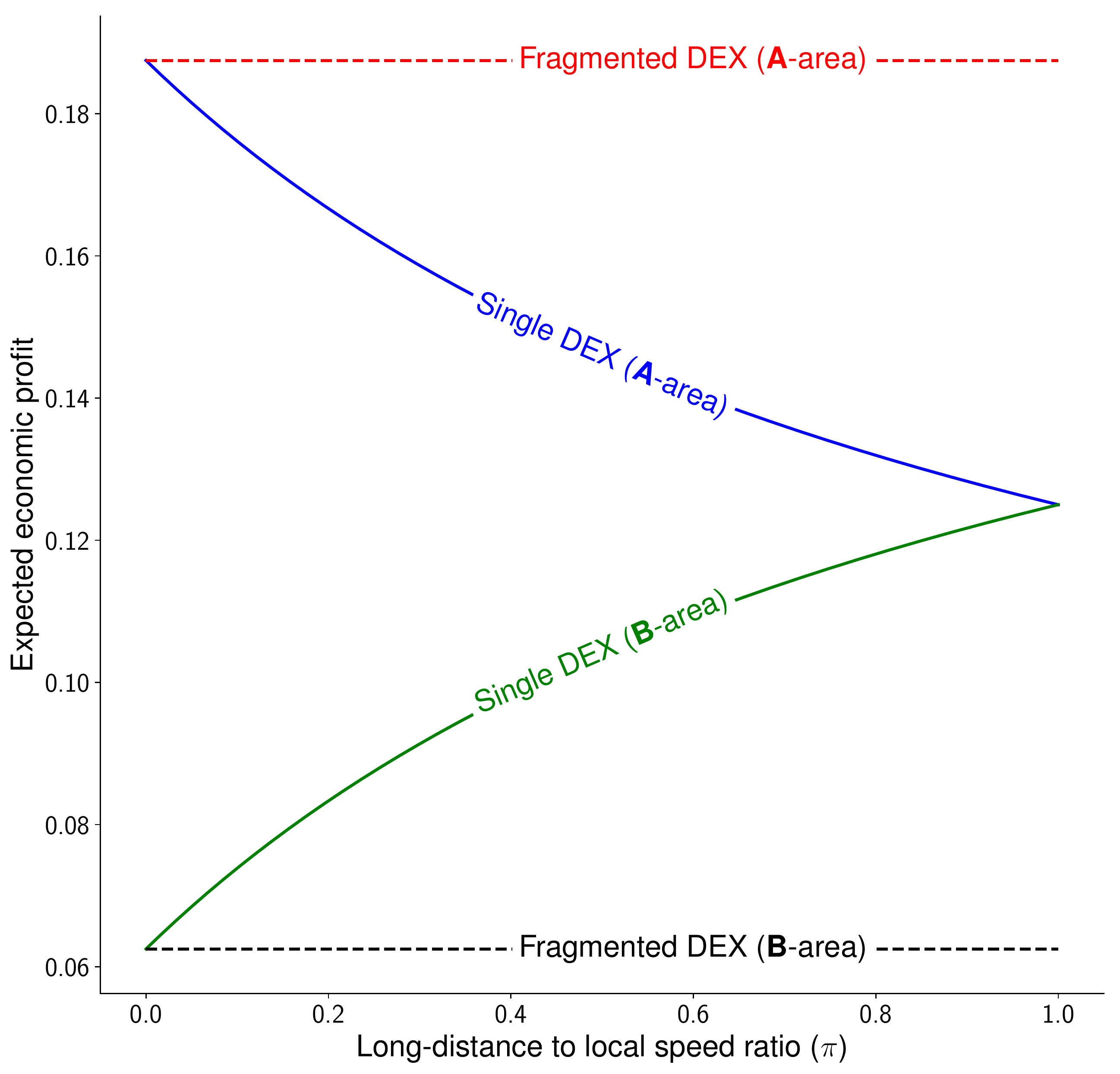}
\par\end{centering}
\end{figure}

Finally, investments in speed can increase fragmentation by boosting the local-trade advantage. Miners have private incentives to increase $\lambda$ and increase their likelihood to capture order flow. If the difference between the expected local trade processing time ($\dfrac{1}{\lambda}$) and the expected long-distance processing time ($\dfrac{1}{\pi\lambda}$) is due to a fixed network latency $\theta$, we can compute $\pi$ (a speed ratio) as a function of network latency $\theta$ as follows:
\begin{equation}\label{eq:theta}
    \frac{1}{\lambda}+\theta = \frac{1}{\pi\lambda} \Rightarrow \pi = \frac{1}{1+\theta\lambda}.
\end{equation}
From equation \eqref{eq:theta}, $\pi$ decreases in $\theta$: if miners in both areas become faster at the same time, the network latency is relatively more important, leading to de facto market fragmentation.

Section \ref{sec:CPU} focuses on estimating $\pi$ as a function of the distance between areas and the asymmetry between the number of nodes in each cluster.

\section{Monte Carlo network experiment \label{sec:CPU}}

\subsection{CPU considerations and network topology \label{sec:topology}}

Transactions on distributed exchanges (seen as a peer-to-peer computer network) are registered in a distributed database (Blockchain) as in Figure \ref{fig:blockchain_alg}.

\begin{figure}[h]
\caption{\label{fig:blockchain_alg} \textbf{Blockchain update algorithm}}
\begin{minipage}[t]{1\columnwidth}%
\footnotesize
$^{(1)}$ set of transactions, hash of previous block, other information (e.g. nonce)
$^{(2)}$ this block may be a very time consuming (e.g for classical proof of work)
$^{(3)}$ a specific consensus algorithm is used and several rounds of messages are interchanged between all invited nodes
\newline
\end{minipage}
\vspace{0.1in}
\begin{centering}
\includegraphics[width=0.7\columnwidth]{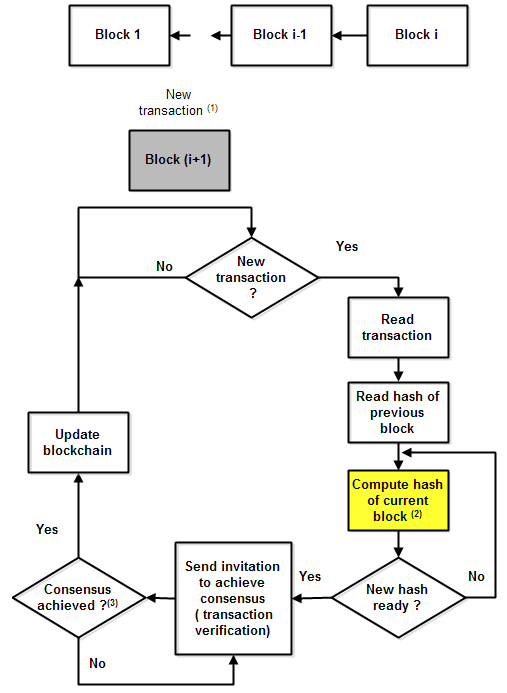}
\par\end{centering}
\end{figure}

Consensus on the distributed exchange is obtained by solving a cryptography (i.e., hashing) problem verified by a quorum with $N$ nodes, where $N$ is sufficiently large. All nodes in the quorum must communicate with each other, to prevent attacks from malicious nodes -- that is, to provide fault tolerance. The time to consensus has two components: the CPU computing time and the speed at which the messages are transferred to the network. The number of messages required depends on the specific consensus algorithm (e.g., $2N^2$ in practical byzantine fault tolerance algorithm or $kN\log_{k}N$ in scalable dynamic multi-agent  practical byzantine fault tolerance where the nodes communicates in groups of $k$ nodes \cite{Consensus_Msg}). 

To reach consensus, all nodes must communicate over several rounds to cross-verify information from other nodes. The consensus is consequently achieved by a qualified, rather than simple, majority of nodes. That is, if we assume the network has $f$ faulty or malicious nodes which provide wrong information, such that $N=3f+1$, consensus emerges when a minimum of $2f+1 > 0.5 N$ nodes validate the transaction.


The verification and consensus achievement protocol, is illustrated in Figure \ref{fig:protocol}, has the following steps:
\begin{enumerate}
    \item Node $i$ (the trader) want to validate a transaction. He sends an \texttt{INVITATION} message to a set of nodes $\mathcal{J}=\left\{j \mid j \text{ can participate to consensus}\right\}$. 
    \item Invited nodes $j\in \mathcal{J}$ reply with an \texttt{ACKNOWLEDGE} message and then receive the transaction information for hash verification.
    \item Each node $j$ submits a \texttt{TRANSACTION OK} message when the hash verification in complete.
    \item A number of messages in several rounds will be exchanged between all nodes (accordingly with a specific consensus algorithm). When the consensus algorithm is complete the transaction is validated and the blockchain is updated for all nodes.
\end{enumerate}

\begin{figure}[h]
\caption{\label{fig:protocol} \textbf{Transaction verification and consensus achievement protocol}}
\begin{minipage}[t]{1\columnwidth}%
\footnotesize
Node $i$ initiates a transaction verification and node $j$ participates to consensus achievement. The messages exchanged between all nodes to achieve consensus are not shown. Nodes $i$ and $j$ may be in  area $A$ or area $B$.
\end{minipage}
\vspace{0.1in}
\begin{centering}
\includegraphics[width=0.7\columnwidth]{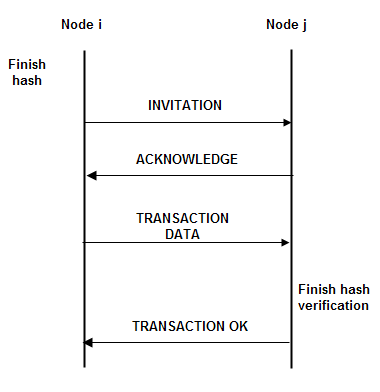}
\par\end{centering}
\end{figure}

The performance of whole system is given by its two components: the processor (i.e., the CPU computing power) and the network (i.e., its topology and the message delay between nodes). We turn next to each of the two components.

Let $\tau_k$ denote the random network delay of one such message, and $\xi_i$ denote the random processing time for miner $i$ to implement a general consensus algorithm (e.g., proof-of-work, proof-of-stake). The total processing time for miner $i$ in region $k$ can be written as the sum of random variables:
\begin{equation}
    T_{i,k}=\frac{1}{2}\tau_{i,k}+\xi_i+\frac{1}{2}\tau_{i,k}=\xi_i+\tau_{i,k}.
\end{equation}
The first miner to complete the algorithm, denoted by $i^{(1)}$, is the one to enter the transaction into the blockchain and collect the fee. Any other node will abort the process. 

In reality, both $\tau$ and $\xi$ are random variables. The randomness in the round-trip time $\tau$ is driven by parameters such as the distance between nodes, asymmetry between clusters, or message length, whereas randomness in $\xi$ emerges from the characteristics of the consensus algorithm. Throughout the rest of the paper, we consider $\xi$ is fixed and disregard variation in processing times. The assumption is motivated by our focus on the impact of geographical distance on miner profits. In particular, the implementation of the consensus algorithm does not depend on the distance from the trader node. Further, competitive pressure generates incentives for each miner to use the best CPU available on the market, which translates into very similar computing times across miners (that is, $\mathbb{E}\xi_i \approx \mathbb{E}\xi_j, \forall i\neq j$). Table \ref{tab:CPU} illustrates, as of October 2019, the highest-performance CPUs available on the retail market \cite{CPU_1, CPU_2, CPU_3, CPU_4, CPU_5}, sorted by their average benchmark score.

\begin{table}[h]
\caption{CPU performance \label{tab:CPU}}
\begin{tabular}{lllll}
\hline
CPU & Launch date & Score range & Average Score\\
\hline
Intel Core i9-9900KF & Q1 2019	&	80-110	& 	101\\
AMD Ryzen 7 3800X& Q3 2019 &	87-100	& 	95.9\\
Intel Core i7-9800X	 & Q4 2018	& 79-101	& 	89.9\\
AMD Ryzen 5 3600& Q3 2019 &  76-95	& 	88.6\\
Intel Core i5-9500& Q2 2019	& 80-92	& 	87.6\\
\hline
\end{tabular}
\end{table}

Further, the network topology features two node clusters, \textbf{A} and \textbf{B}, standing in for two different geographies (e.g., cities where trading activity concentrates). We allow the clusters to have a different number of nodes, to model asymmetry as in Section \ref{sec:econ}. Within each cluster, or area, nodes are connected by high speed links (low delay). Between clusters, nodes are connected by  low speed links (high delay). That is, the distance between areas \textbf{A} and \textbf{B} is modeled through a larger messaging delay. 

\begin{figure}[h]
\caption{\label{fig:net_model} \textbf{Network model}}
\begin{minipage}[t]{1\columnwidth}%
\footnotesize
Nodes are fully interconnected (not all the link are illustrated). Only nodes A0--A3 and B0--B2 participate to achieve consensus. Node A0 or node B0 will initiate transaction verification. 
\end{minipage}
\vspace{0.1in}
\begin{centering}
\includegraphics[width=\columnwidth]{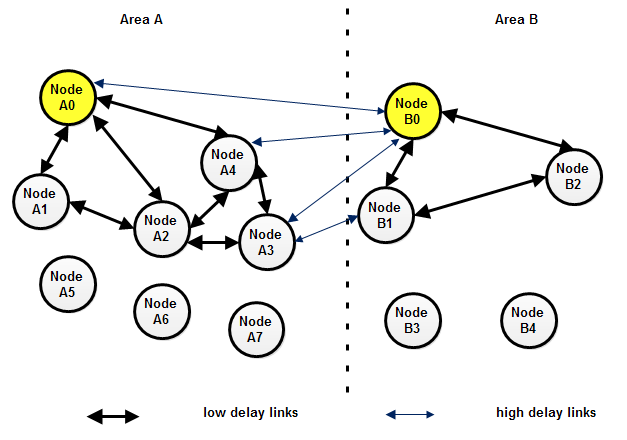}
\par\end{centering}
\end{figure}

Figure \ref{fig:net_model} illustrates such a network topology. The node initiating the transaction (and collecting the fee) can be located either in area \textbf{A} or in area \textbf{B} -- without loss of generality, we label it as $A_0$ or $B_0$. The Figures also shows network nodes that do not participate in the consensus quorum for the given transaction.

Throughout the rest of the paper, we consider a network with 10 nodes, divided in two clusters \textbf{A} and \textbf{B}. We consider different cluster configurations, indexed by $\left(\eta_A,\eta_B\right)$, where $\eta_A$ and $\eta_B$ are the number of nodes in clusters \textbf{A} and \textbf{B}, respectively, with $\eta_A+\eta_B=10$. Without loss of generality, we fix cluster $A$ to be the larger one and consider $\eta_A \supset \left\{5,6,7,8,9\right\}$. The case of $\eta_A=5$ corresponds to identically-sized clusters; the case of $\eta_A=10$ is ignored as it corresponsds to the degenerate case of a single cluster. Fast, intra-cluster, links have a latency between 10 and 30 milliseconds across all specifications. We consider six values for the average latency of a slow link, ranging from 50 to 300 milliseconds in 50 ms increments.

The round trip time (delay) distribution is estimated using a simulation program written in \texttt{C} and developed in the Visual Studio software. That is, we simulate communication in a  unstructured peer-to-peer network.  The simulator randomly selects a source and destination from the $N$ nodes in the quorum, then draws a random route between the source and the destination. The message delay is computed as a sum of delays between hops (i.e., intermediary nodes on the route route).  The communication delays between nodes are uniform random variables in a predefined interval different for fast links (inter cluster) and slow links (intra cluster).

\subsection{Estimation of delay distribution \label{sec:marginal}}

Let $F_\ell\left(\tau\right)$ be the cumulative distribution function of miner message delay $\tau$, where $F_\ell\left(x\right)\equiv\mathbb{P}\left(\tau\leq x\right)$. The parameter $\ell$ is a set of conditioning variables (distance, number of nodes active, etc.).

We approximate the (theoretical) cumulative distribution function $F\left(\cdot\right)$ with its empirical counterpart obtained through network simulations, that is 
\begin{equation}\label{eq:empCDF}
    \hat{F}_\ell\left(x\right)=\frac{\text{Number of observations smaller than } x}{\text{Number of observations } (N)}
\end{equation}
where $\hat{F}_\ell\left(x\right)$ converges to $F_\ell\left(x\right)$ if $N$ is arbitrarily large. For each parameter combination, we set $N=100,000$ (one hundred thousand simulations).

We estimate the empirical cumulative distribution for the entire network ($F^\text{full}_\ell\left(\cdot\right)$) as well as conditional on the source node being in cluster $A$ or $B$ ($F^\text{A}_\ell\left(\cdot\right)$ and $F^\text{B}_\ell\left(\cdot\right)$, respectively). Section \ref{sec:networkwide} focuses on network-wide delays, whereas in Section \ref{sec:winningprob} we conduct a Monte Carlo experiment to determine the relative network delays of nodes across clusters.

\subsection{Network-wide results \label{sec:networkwide}}

We estimate the average network delay to consensus using a standard bootstrap algorithm: from the full sample of 100,000 simulated delays we draw 1,000 sub-samples of 5,000 observations and subsequently average their means.  Figure \ref{fig:average_consensus} suggests that the average network delay to consensus following a transaction decreases both (a) in the latency of inter-cluster links and (b) in the cluster asymmetry level. 

The first result is intuitive: the closer the clusters are, geographically, the faster computers communicate across cluster. The second result is driven by the fact that, if clusters are more asymmetric, the route between the source and destination node is less likely to contain inter-cluster, high-delay hops and more likely to contain intra-cluster, low-latency hops. As a result, the average delay is lower in more asymmetric networks.


Interestingly, the variability of the delay decreases in the level of asymmetry. The result follows immediately from the higher delay variability of slow links, which are less likely to be used for consensus in an asymmetric network topology.

\begin{figure}[h]
\caption{\label{fig:average_consensus} \textbf{Average message delay (full network)}}
\begin{minipage}[t]{1\columnwidth}%
\footnotesize
This figure illustrates the mean message delay in a network with 10 computers. The cluster structure is on the x-axis, from an equal distribution of five computers per cluster $(5,5)$ to the most uneven distribution with nine computers in a single cluster $(9,1)$. Different lines correspond to different delays on low-speed links. Means are bootstrapped via 1,000 draws of 5,000 observations from the original population of 100,000 delays. 
\end{minipage}
\vspace{0.2in}
\begin{centering}
\includegraphics[width=\columnwidth]{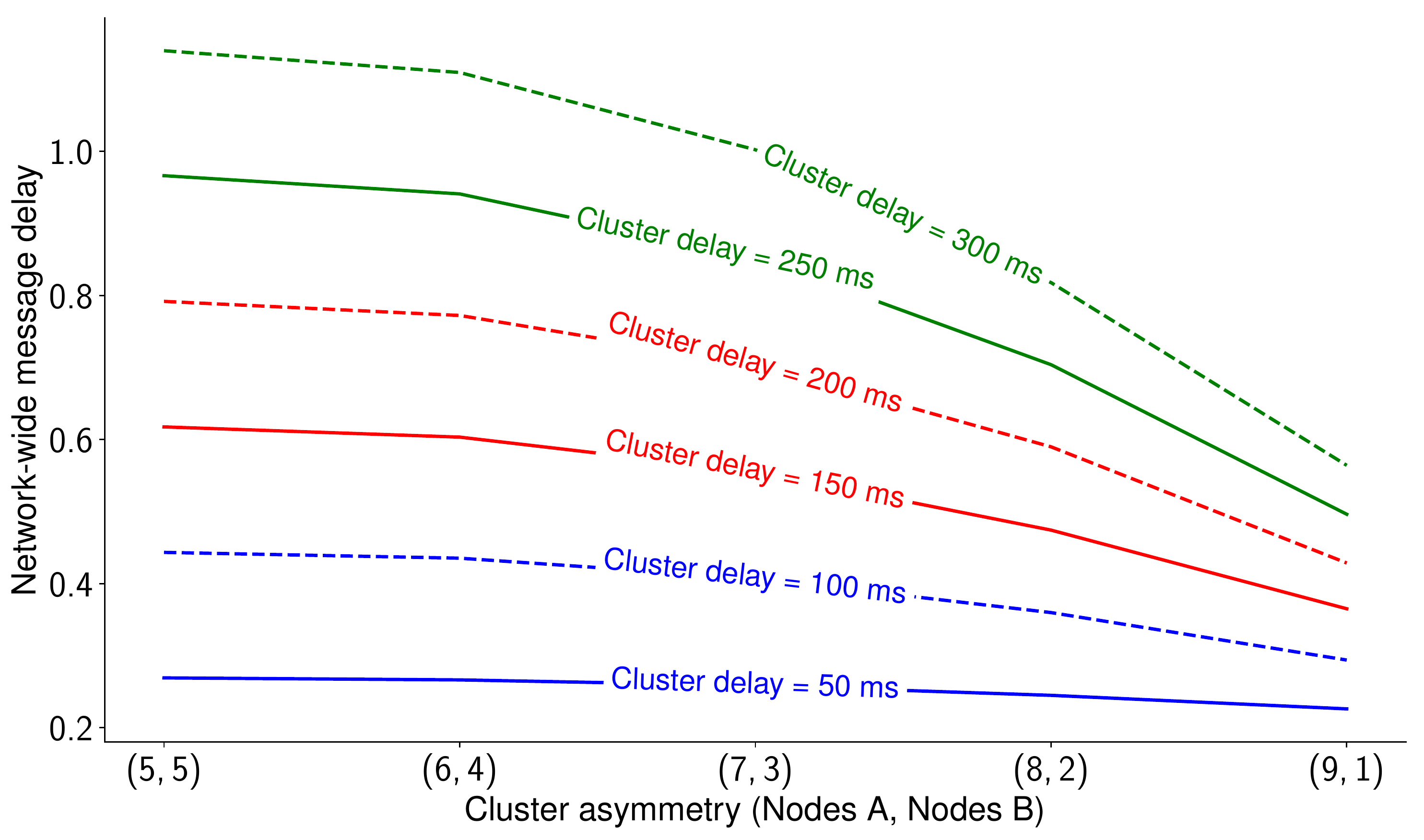}
\par\end{centering}
\end{figure}

Figure \ref{fig:prob_distribution} illustrates the density and cumulative distribution of delays, using the full samples of 100,000 observations per topology. We note that more symmetric network topology configurations exhibit ``fatter tails,'' that is a higher probability of extreme delays to consensus. The ``fat tail'' effect is even more pronounced for relatively slower inter-cluster links, as variability increases.

\begin{figure}[h]
\caption{\label{fig:prob_distribution} \textbf{Probability distribution of message delays}}
\begin{minipage}[t]{1\columnwidth}%
\footnotesize
This figure illustrates the empirical (simulated) distribution of the message delay in a network with 10 computers, for different cluster asymmetry levels and link speeds. We plot probability densities (left) and cumulative distribution functions (right) for a inter-cluster delay of 50 ms (top panel) and 300 ms respectively (bottom panel)
\end{minipage}
\vspace{0.2in}

\begin{centering}
\subfloat[Inter-cluster average delay is 50 ms.]{
\begin{centering}
\includegraphics[width=0.45\columnwidth]{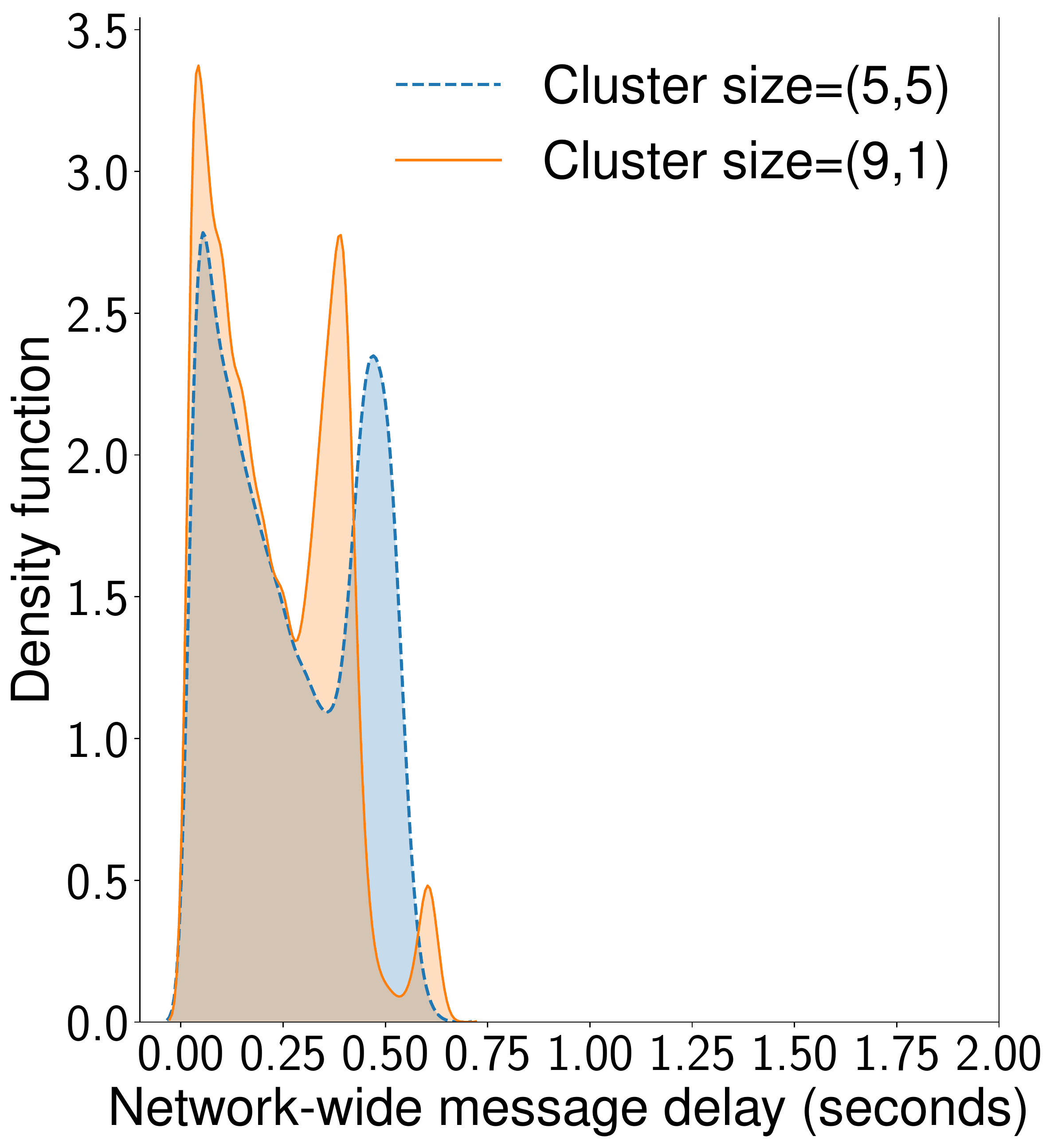}\includegraphics[width=0.45\columnwidth]{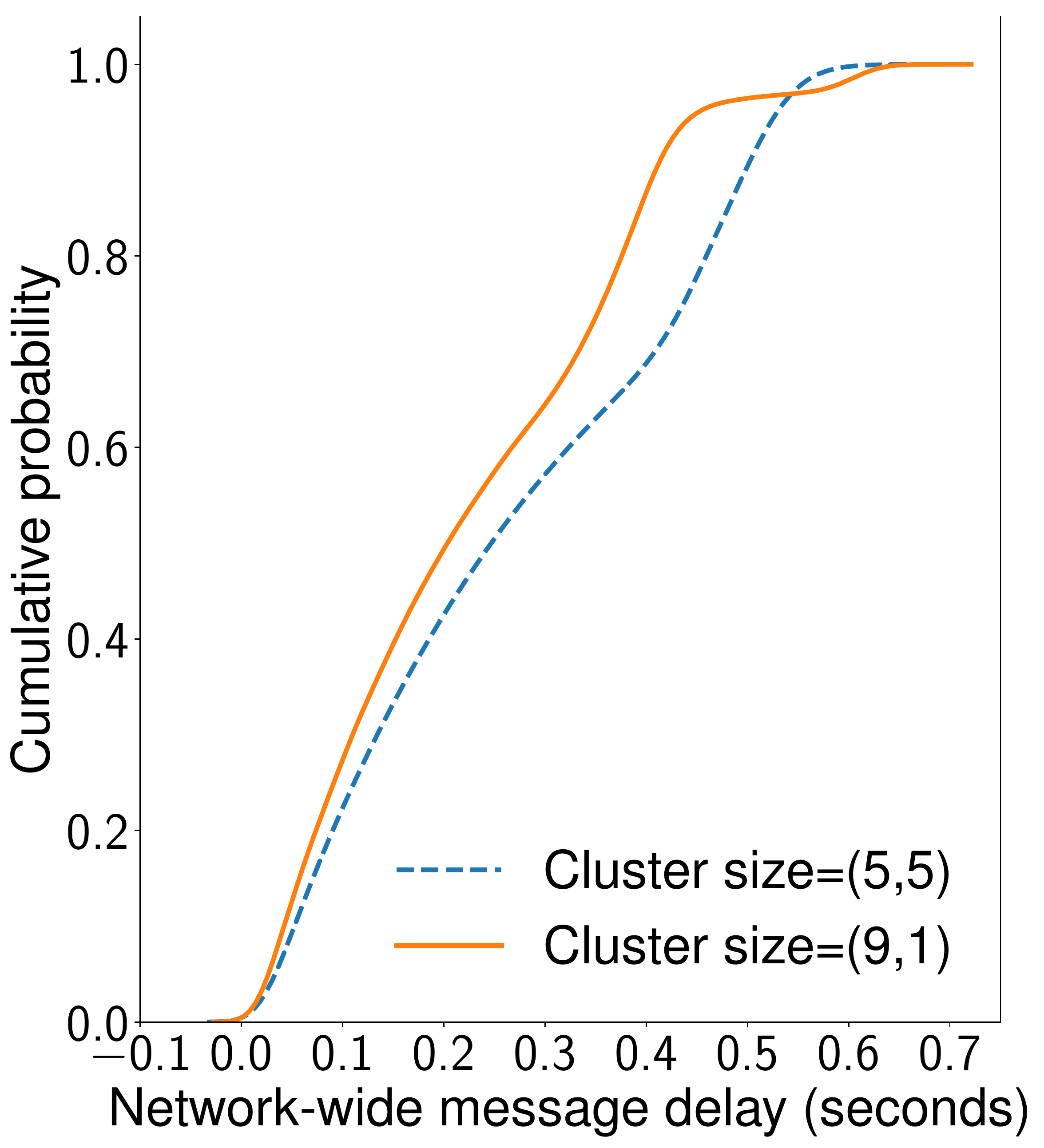}
\end{centering}}
\end{centering}

\begin{centering}
\subfloat[Inter-cluster average delay is 300 ms.]{
\begin{centering}
\includegraphics[width=0.45\columnwidth]{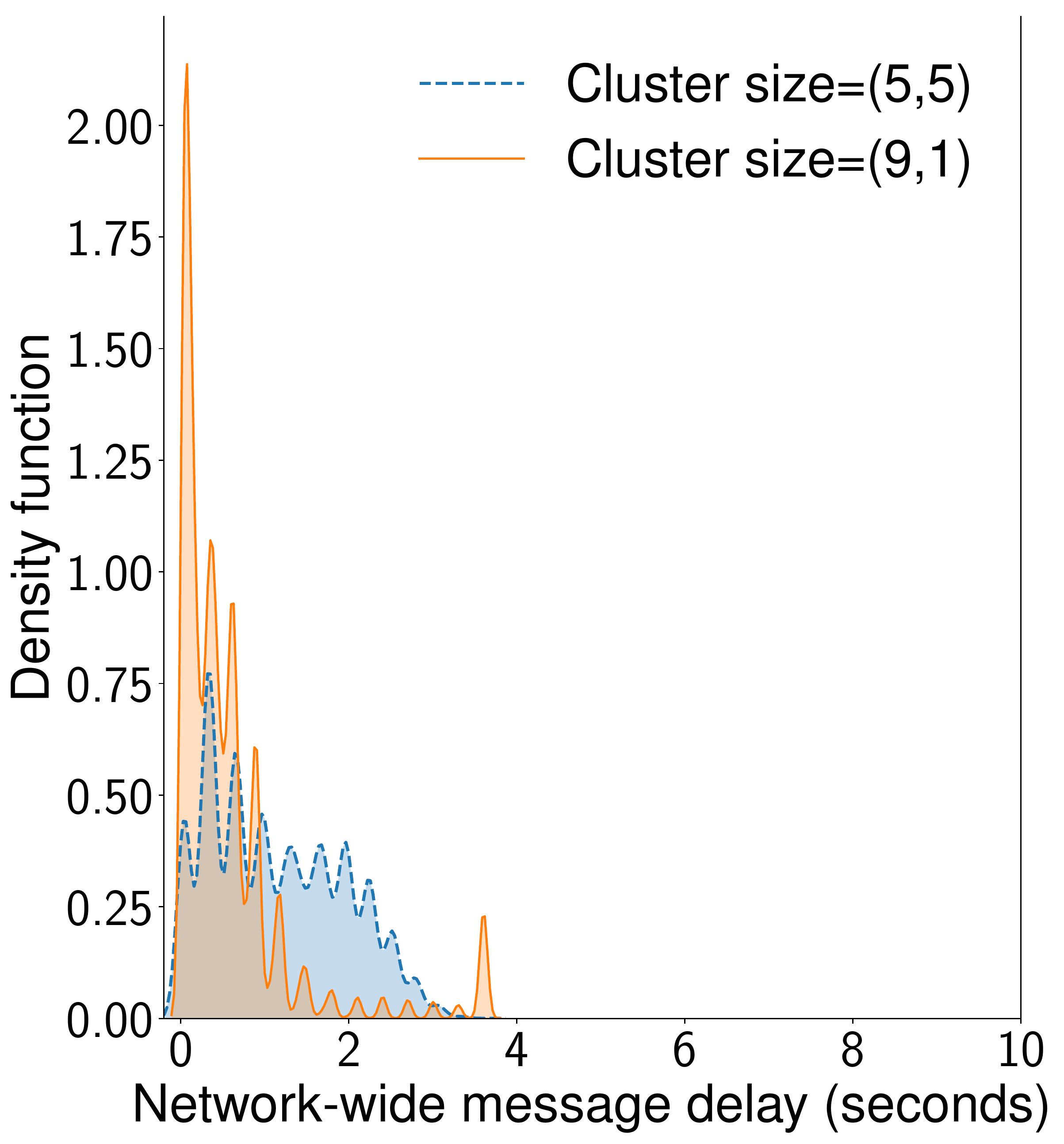}\includegraphics[width=0.45\columnwidth]{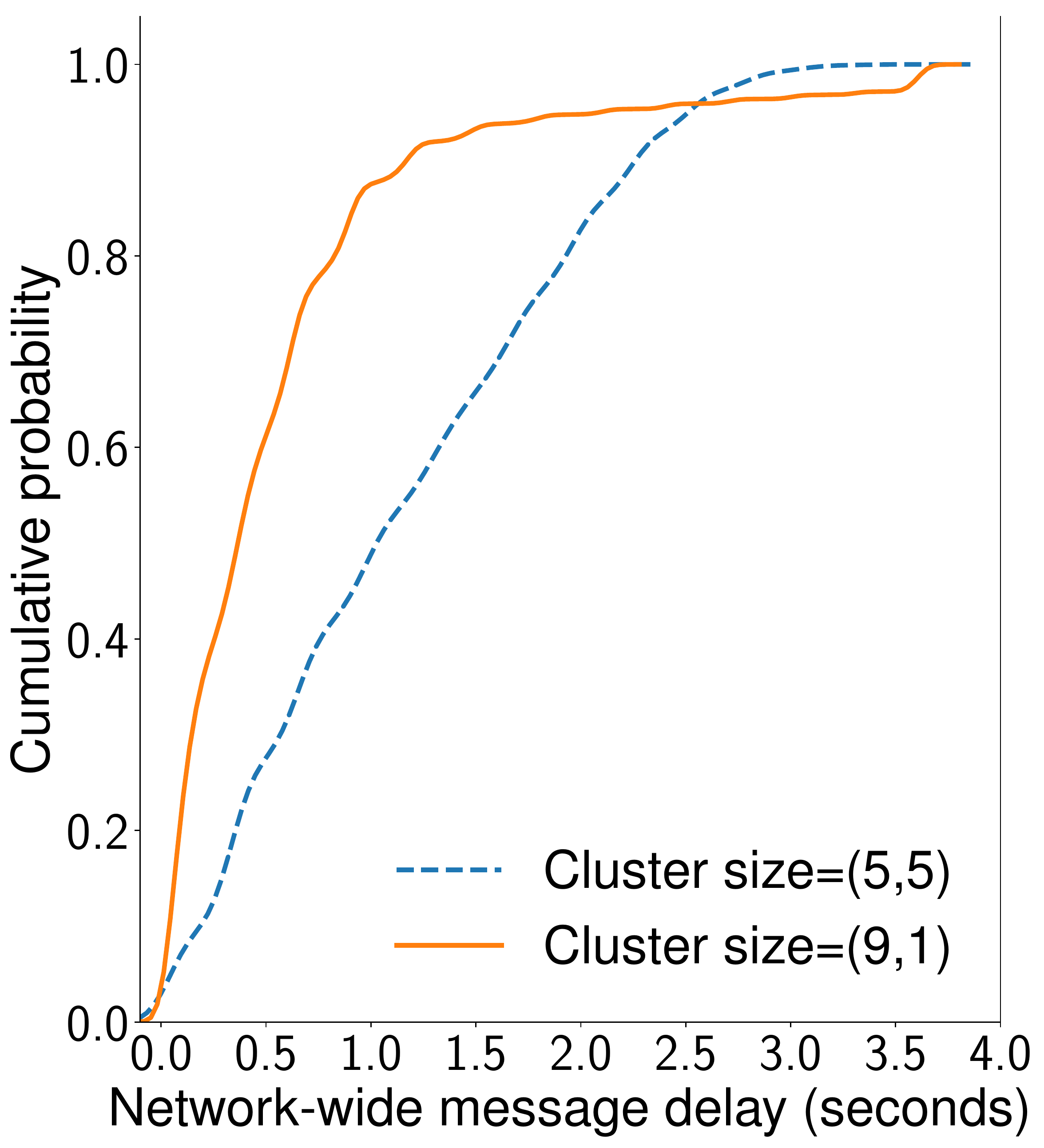}
\end{centering}}
\end{centering}

\end{figure}

\subsection{Large cluster speed advantage \label{sec:winningprob}}

In this section, we consider delays for nodes in clusters \textbf{A} and \textbf{B} separately. We implement a Monte Carlo experiment to estimate the relative probability of posting a transaction to the Blockchain for nodes in either cluster (i.e., $\pi$ in Section \ref{sec:econ}), as a function of (i) the asymmetry and (ii) the distance between clusters as defined in Section \ref{sec:topology}.

As a first step, we estimate 100,000 delays for nodes in each cluster and estimate the empirical distribution functions of the delay using equation \eqref{eq:empCDF}, that is $\hat{F}^A_{\ell,\left(\eta_A,\eta_B\right)}$ and $\hat{F}^B_{\ell,\left(\eta_A,\eta_B\right)}$. Figure \ref{fig:prob_distribution_nodes} displays the probability distributions of delays for both clusters, for different cluster asymmetry levels and inter-cluster delays. 

\begin{figure}[h]
\caption{\label{fig:prob_distribution_nodes} \textbf{Probability distribution of message delays in each cluster}}
\begin{minipage}[t]{1\columnwidth}%
\footnotesize
This figure illustrates the empirical (simulated) distribution of the message delay in a network with 10 computers, for each cluster $A$ and $B$. The inter-cluster delay is fixed to 300 ms. The top (bottom) panel displays probability distribution functions for low (high) asymmetry levels. 
\end{minipage}
\vspace{0.2in}

\begin{centering}
\subfloat[Asymmetry levels are: \emph{left} (5,5) and \emph{right} (6,4) ]{
\begin{centering}
\includegraphics[width=0.45\columnwidth]{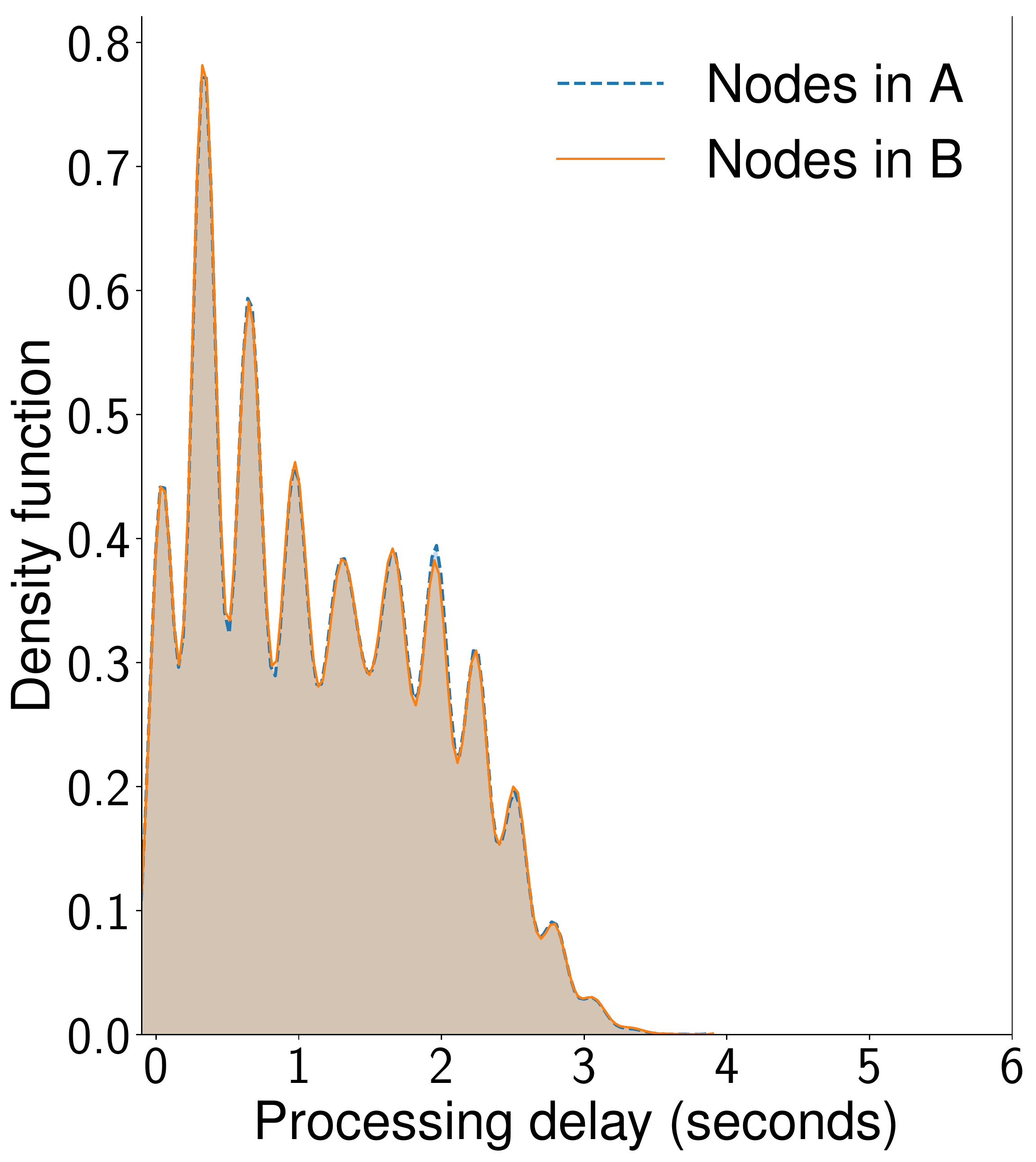}\includegraphics[width=0.45\columnwidth]{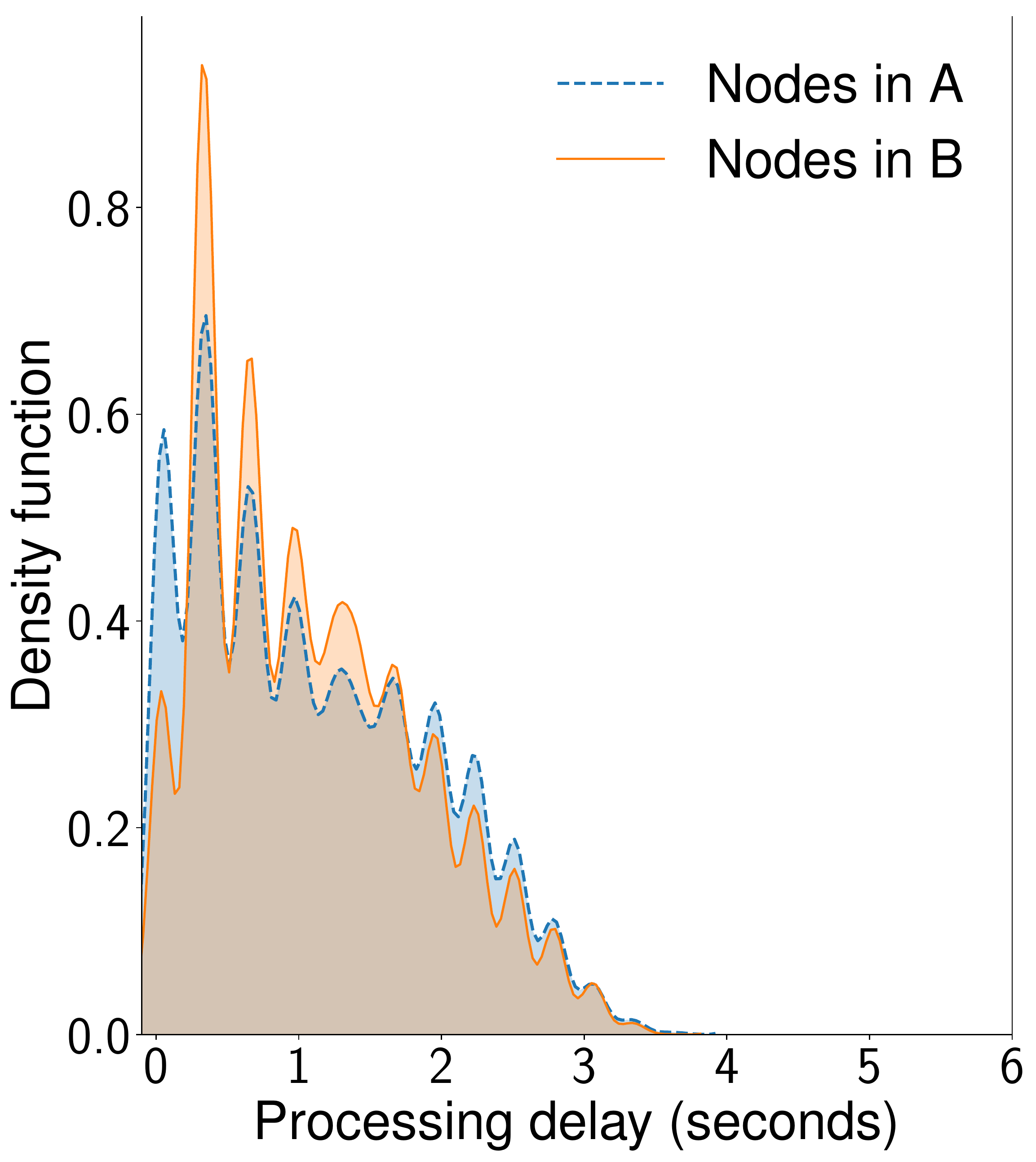}
\end{centering}}
\end{centering}

\begin{centering}
\subfloat[Asymmetry levels are: \emph{left} (8,2) and \emph{right} (9,1)]{
\begin{centering}
\includegraphics[width=0.45\columnwidth]{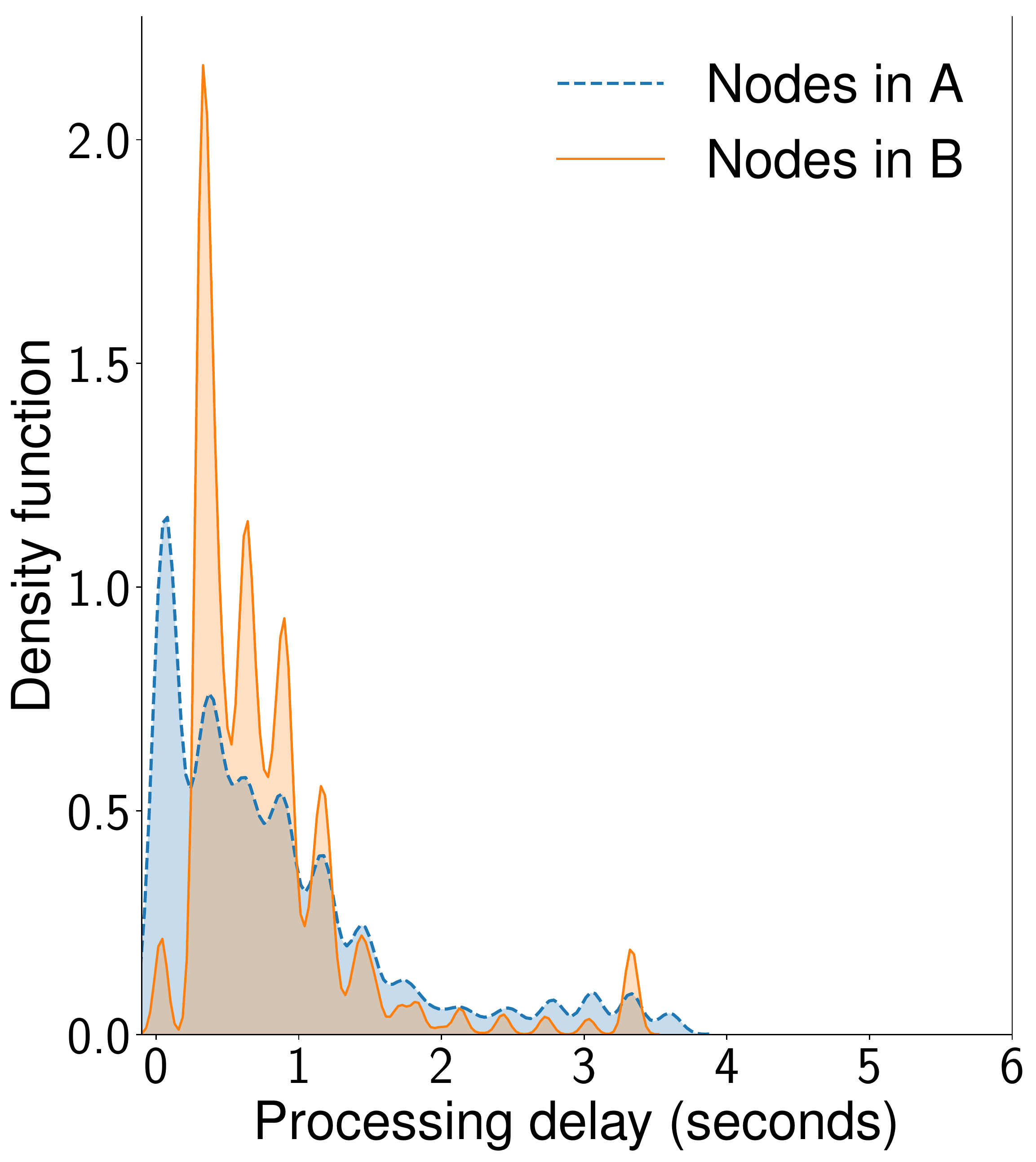}\includegraphics[width=0.45\columnwidth]{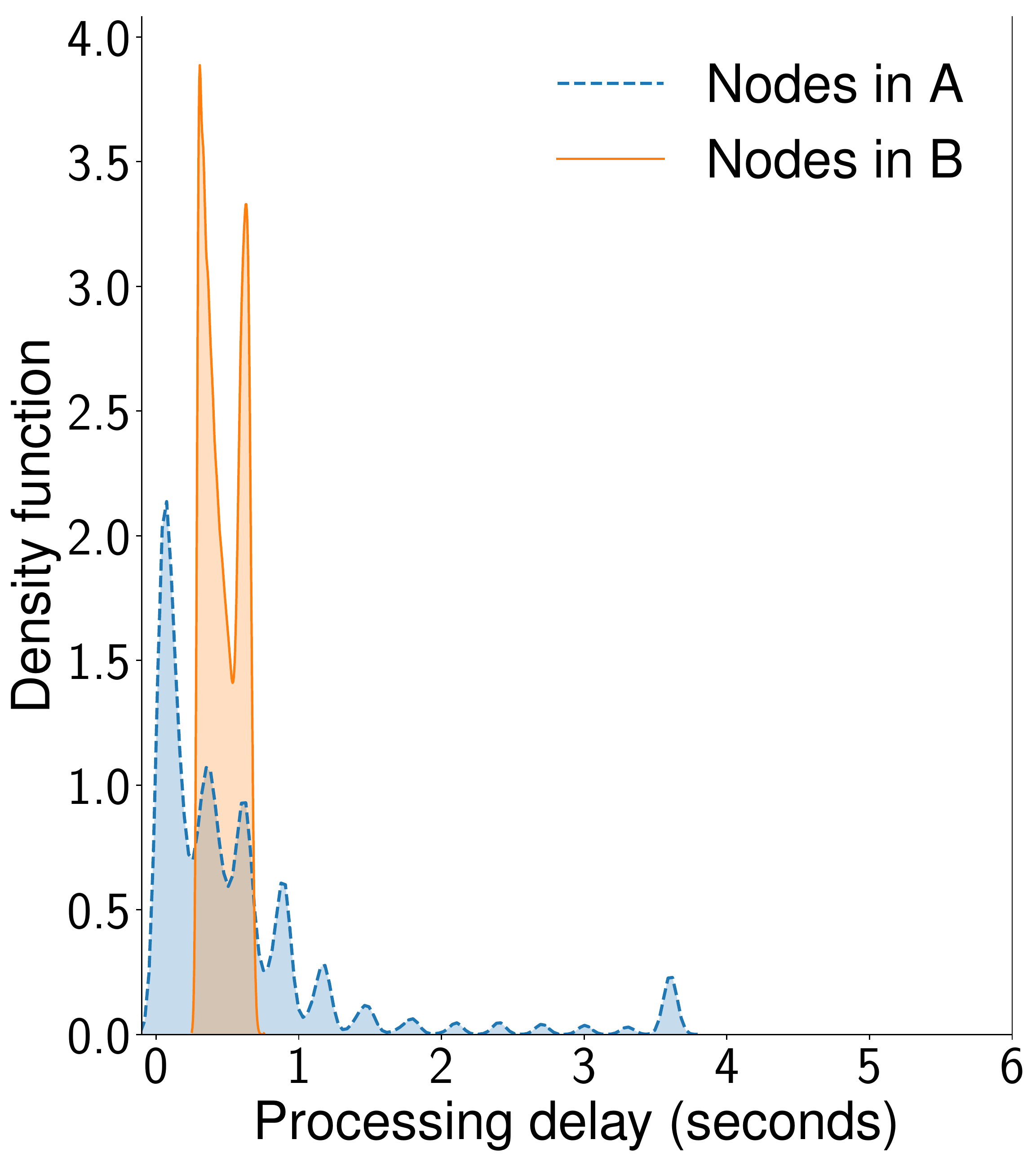}
\end{centering}}
\end{centering}

\end{figure}

We compare the message delay distributions for nodes in \textbf{A} and \textbf{B} regions, respectively. First off, the distributions perfectly overlap if there is no asymmetry between clusters -- that is, if each region contains five nodes. The result is natural since the network topology is symmetrical. As clusters become more asymmetric, the mean delay decreases for nodes in both \textbf{A} and \textbf{B}.  The intuition is that if clusters are more asymmetric, there are relatively fewer routes between any two nodes that include inter-cluster links. That is, messages are less likely to travel back and forth between clusters to reach their destination, leading to shorter delays on average. At the same time, the delay for nodes in $A$ becomes more left-skewed. Even if the \emph{average} delay for a node in \textbf{A} is not necessarily smaller than the average delay for a node in \textbf{B}, there is a higher probability of a very small delay -- probability which increases with cluster asymmetry.

Once the empirical distribution is estimated, the Monte Carlo algorithm proceeds as follows, for each of $n_\text{sim}=1,000$ simulations:
\begin{enumerate}
            \item We draw $n_\text{draws}=10,000$ network delays for each node. That is, for each draw we take $\eta_A$ values from $\hat{F}^A_{\ell,\left(\eta_A,\eta_B\right)}$ and $\eta_B$ values from $\hat{F}^B_{\ell,\left(\eta_A,\eta_B\right)}$. 
            \item We compute the minimum over the $\eta_A$ and $\eta_B$ delays, and label it as $A$ or $B$, depending on its originating cluster.
            \item We compute the probability of cluster $A$ having the smaller delay in a particular simulation, as:
            \begin{equation}
                \hat{p}_{A,k} = \frac{\#\text{ times smallest delay $\in A$}}{n_\text{draws}},
            \end{equation}
            where $k$ runs over simulations.
    \item We average the estimated probabilities $\hat{p}_{A,k}$ over the $n_\text{sim}=1,000$ simulations
        \begin{equation}
            \hat{p}^\text{cluster}_{A}=\frac{\sum_{k}\hat{p}_{A,k}}{n_\text{sim}}.
        \end{equation}
\end{enumerate}

Figure \ref{fig:win_cluster} illustrates the Monte Carlo results at cluster level. Nodes in the larger cluster (i.e., cluster $A$) are on aggregate more likely to have the smallest message delay and ``win'' the transaction if either the cluster size asymmetry is larger or clusters are further apart geographically.

\begin{figure}[h]
\caption{\label{fig:win_cluster} \textbf{Probability of smallest delay (cluster level)}}
\begin{minipage}[t]{1\columnwidth}%
\footnotesize
This figure illustrates the probability, at the cluster level, that \emph{any} node in the given cluster has the smallest delay to post a transaction on the Blockchain.
\end{minipage}
\vspace{0.2in}
\begin{centering}
\includegraphics[width=\columnwidth]{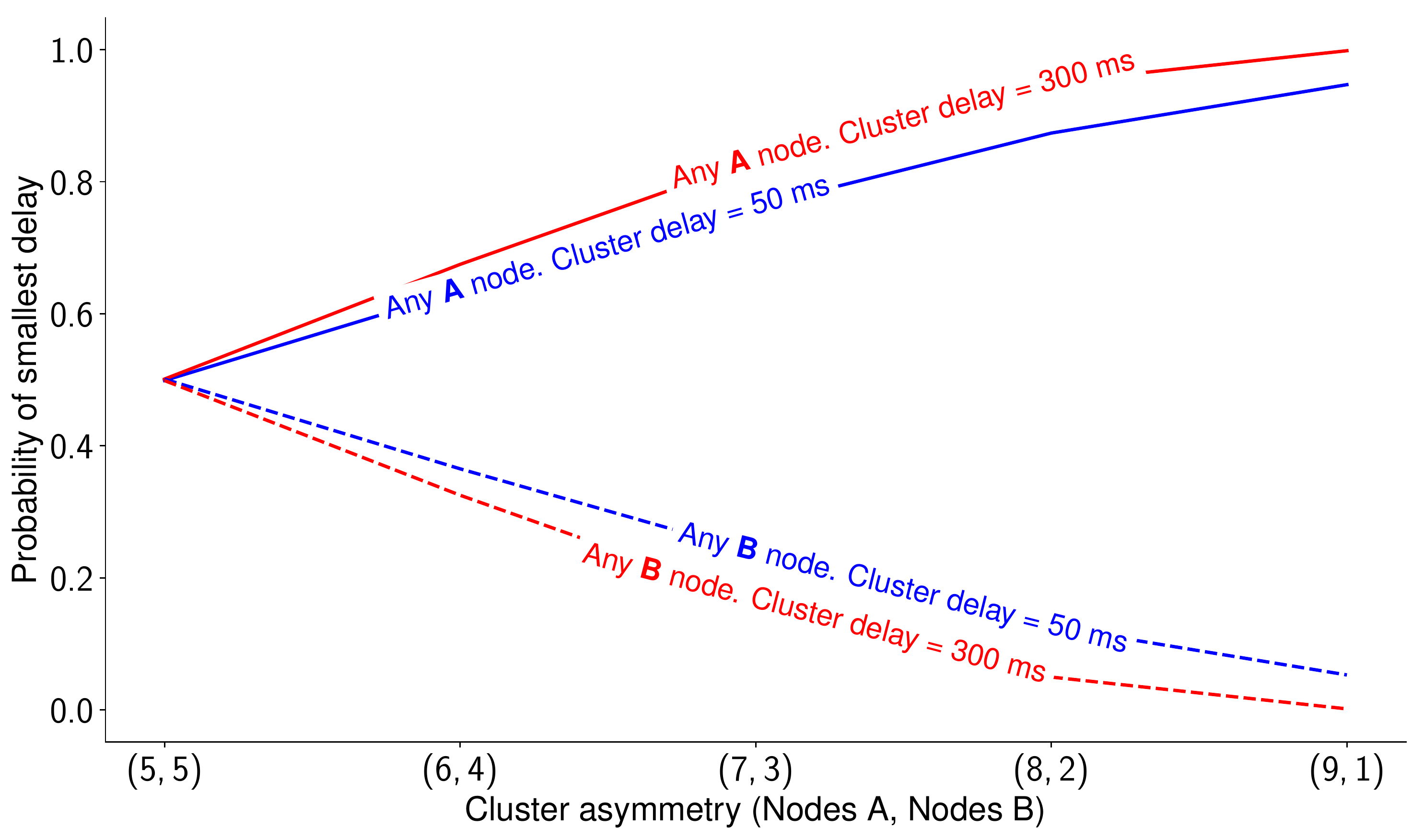}
\par\end{centering}
\end{figure}

However, since the study is focused on individual nodes' incentives to join a particular blockchain, we need to estimate the probability of an \emph{individual} node in a given cluster having the smallest delay, that is $\hat{p}^\text{node}_{j}$ with $j\in\left\{A,B\right\}$. Since nodes in each group are independent, we estimate the individual probabilities as
\begin{equation}
    \hat{p}^\text{node}_{A}=\frac{\hat{p}^\text{cluster}_A}{\eta_A} \text{ and } \hat{p}^\text{node}_{B}=\frac{1-\hat{p}^\text{cluster}_A}{\eta_B}.
\end{equation}

Figure \ref{fig:win_individual} shows that as cluster asymmetry increases, each node is eventually less likely to arrive first. Even if nodes in cluster $A$ are \emph{on aggregate} favored by the asymmetry, the probability of any node in cluster $A$ ``winning'' increases at a slower rate than the number of nodes. We obtain a concave relationship between $\hat{p}^\text{node}_{A}$ and cluster asymmetry: $\hat{p}^\text{node}_{A}$ first increases, then deceases in the level of asymmetry. For nodes in cluster $B$, higher asymmetry unambiguously decreases the probability of having the smallest delay.

\begin{figure}[h]
\caption{\label{fig:win_individual} \textbf{Probability of smallest delay (individual node)}}
\begin{minipage}[t]{1\columnwidth}%
\footnotesize
This figure illustrates the probability, at the cluster level, that \emph{a particular} node in the given cluster has the smallest delay to post a transaction on the Blockchain.
\end{minipage}
\vspace{0.2in}
\begin{centering}
\includegraphics[width=\columnwidth]{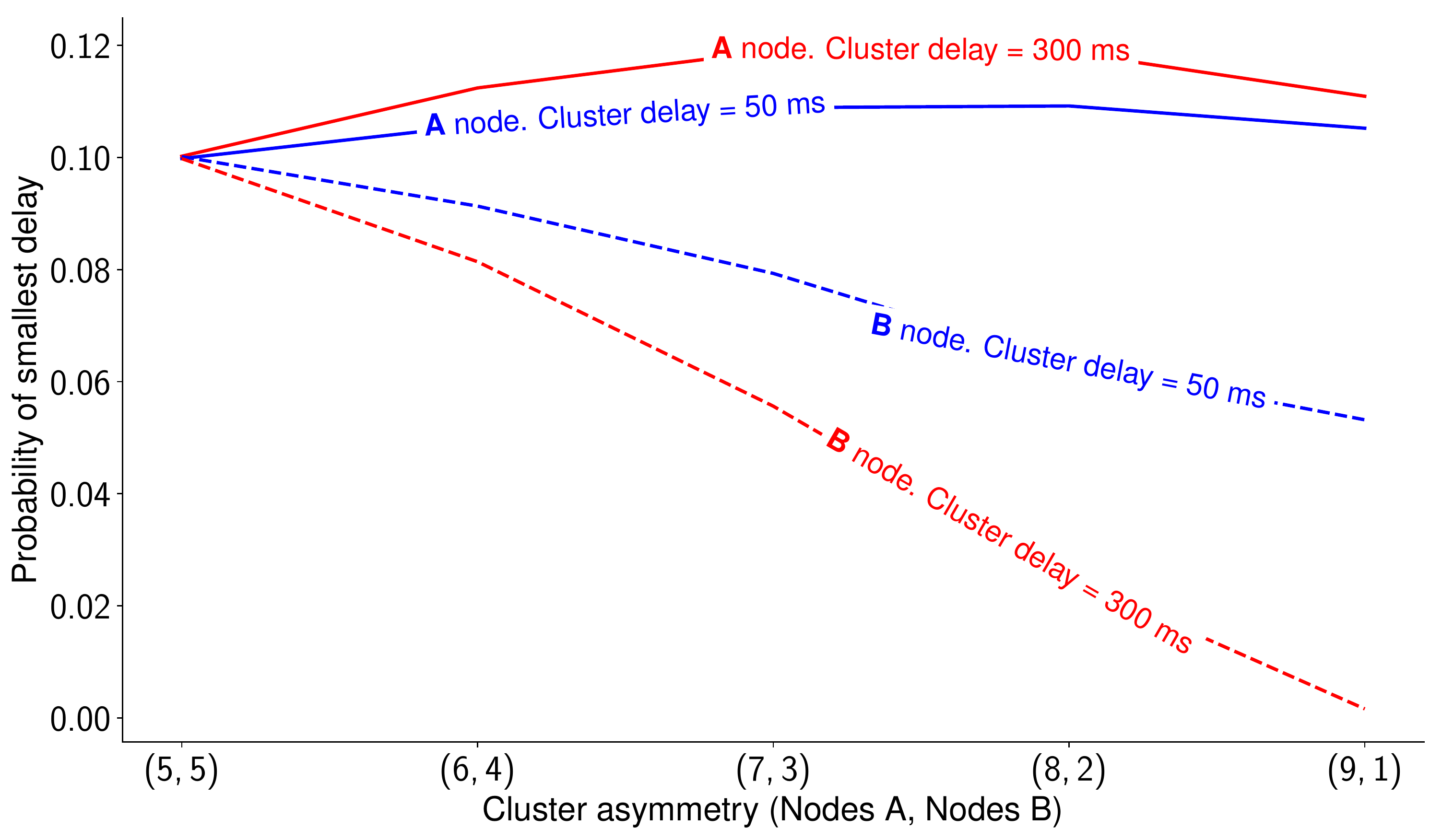}
\par\end{centering}
\end{figure}

Figure \ref{fig:win_lr} illustrates the relative advantage of a node in cluster $A$ -- that is, a measure of $\pi$ from Section \ref{sec:econ}. The relative advantage increases with both cluster distance and cluster asymmetry. Moreover, the two effects reinforce each other: asymmetry favors the larger cluster more when the distance between clusters increases. At its peak (for an inter-cluster delay of 200 ms and a single computer out of ten in cluster $B$), a node in cluster $A$ is 30 times more likely to have the shortest messaging delay and win the transaction.

\begin{figure}[h]
\caption{\label{fig:win_lr} \textbf{Winning likelihood ratio}}
\begin{minipage}[t]{1\columnwidth}%
\footnotesize
This figure illustrates the ratio between the probability of smallest delay for a node in cluster \textbf{A} and a node in cluster \textbf{B}, respectively.
\end{minipage}
\vspace{0.2in}
\begin{centering}
\includegraphics[width=\columnwidth]{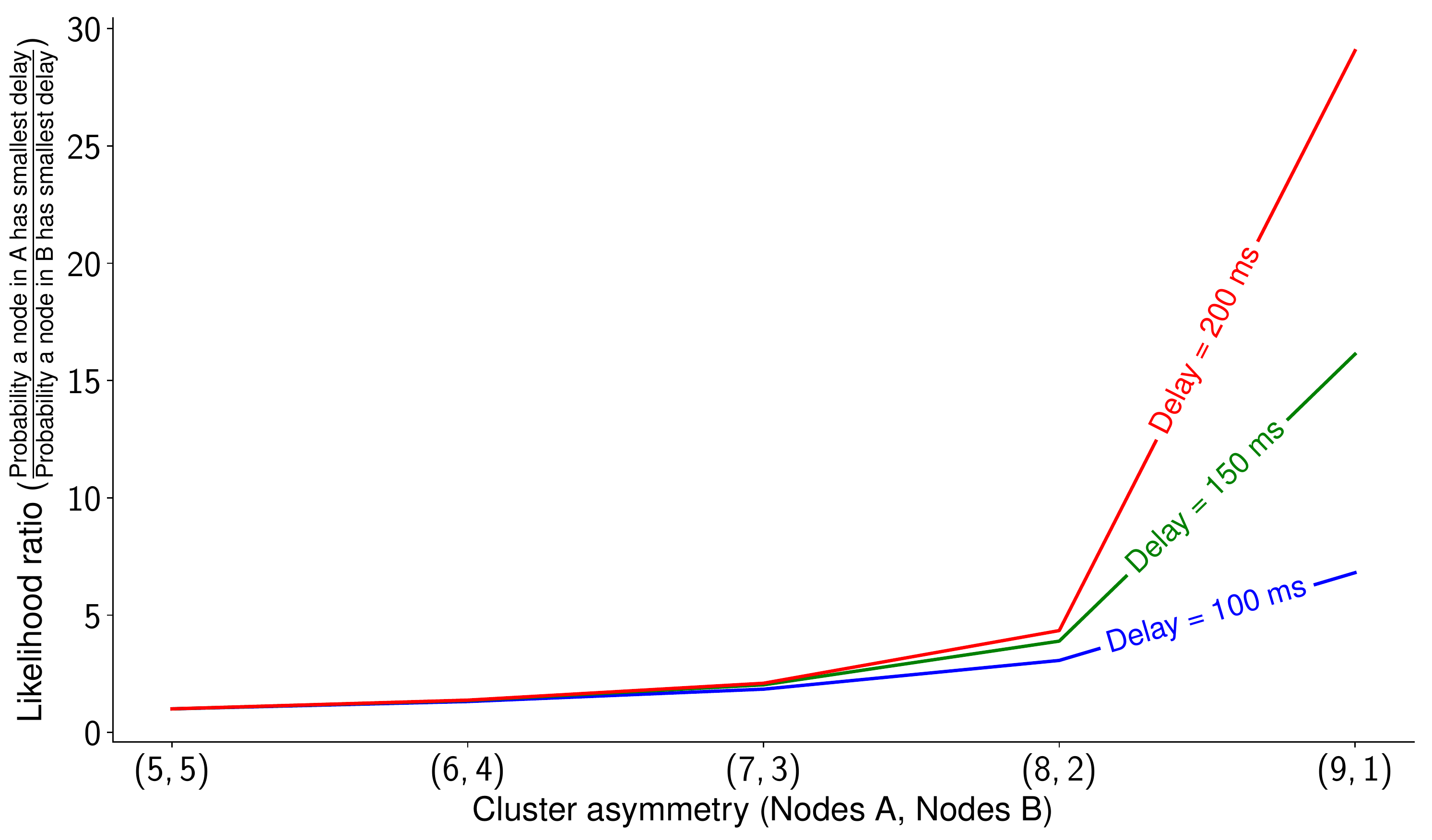}
\par\end{centering}
\end{figure}

Formally, we estimate a linear regression model to capture the relationship between winning probabilities and cluster characteristics:
\begin{equation}
    \textbf{y}= \beta_0+\beta_1 \text{Delay}+\beta_2 \frac{\eta_A}{\eta_B} + \beta_3 \text{Delay} \times  \frac{\eta_A}{\eta_B} +\text{error},
\end{equation}
where the dependent variable $\textbf{y}$ can take four values, $$\textbf{y}\in\left\{\hat{p}_A^\text{cluster},\hat{p}_A^\text{node},\hat{p}_B^\text{node},\dfrac{\hat{p}_A^\text{node}}{\hat{p}_B^\text{node}}\right\},$$ \emph{Delay} stands in for the inter-cluster delay and $\dfrac{\eta_A}{\eta_B}$ is a measure of cluster asymmetry (standardized in the regression to have zero mean and unit standard deviation). Table \ref{tab:tech} presents the results of the regression. In particular, the speed advantage of the large cluster (the ratio $\dfrac{\hat{p}_A^\text{node}}{\hat{p}_B^\text{node}}$ in the last column) increases in both the inter-cluster delay and cluster asymmetry. There is a positive and significant interaction effect between the two. The effects are primarily driven from a sharp decrease in the winning probability of nodes in \textbf{B} (third column) rather than an increase in the winning probability of nodes in \textbf{A} (second column).

\begin{table}
\caption{\label{tab:tech} \textbf{Determinants of ``local'' speed advantage}}
    \centering
\vspace{0.1in}
\begin{tabular}{lllll} 
\toprule
& \multicolumn{4}{c}{Dependent variable ($\textbf{y}$)} \\
& $\hat{p}_A^\text{cluster}$ & $\hat{p}_A^\text{node}$ & $\hat{p}_B^\text{node}$ & $\dfrac{\hat{p}_A^\text{node}}{\hat{p}_B^\text{node}}$ \\
\cmidrule{2-5}
Constant ($\beta_0$) & 77.96*** & 11.06*** & 5.87*** & 7.19*** \\
 & (41.02) & (98.75) & (22.30) & (8.53) \\
Inter-cluster delay ($\beta_1$)  & 1.41 & 0.19* & -0.72*** & 4.84*** \\
 & (0.76) & (1.79) & (-2.75) & (4.50) \\
Cluster asymmetry ($\beta_2$) & 14.68*** & 0.11 & -2.90*** & 10.23***\\
 & (9.44) & (1.187) & (-11.54) & (10.31) \\
Delay $\times$ asymmetry ($\beta_3$)  & 0.27 & 0.02 & -0.50 & 8.73*** \\
 & (0.16) & (0.15) & (-1.51) & (6.37) \\
\cmidrule{1-5}
Model $R^2$ & 66\% & 11\% & 81\% & 90\% \\
\bottomrule
\multicolumn{5}{c}{Asymptotic norrmal z-statistics in parentheses.} \\
\multicolumn{5}{c}{Heteroskedasticity-robust standard errors. *** p$<$0.01, ** p$<$0.05, * p$<$0.1.} \\
\multicolumn{5}{c}{A higher $R^2$ indicates the model explains more of the data variation.} \\
\end{tabular}
\end{table}

\section{Conclusion}
This paper studies the limits of distributed exchange (DEX) infrastructure. One potential advantage of a DEX is to provide a unique marketplace for buyers and sellers and maximize network effects, while at the same time eliminating concerns related to the exchange operator's market power. Can such a distributed exchange operate over geographically distant regions (e.g., different cities or countries where securities are cross-listed)? To answer the question, we both build an economic model of DEX and conduct Monte Carlo simultations of unstructured P2P networks.

We find that asymmetry between \emph{economic activity} levels generates a value transfer from trading infrastructure providers in high-activity regions to providers in low-activity regions. Consequently, miners in high-activity regions only have incentives to join a distributed exchange when then benefit from a significant processing speed advantage over the peer-to-peer network. Through P2P network simulations, we find that the speed advantage increases in the level of \emph{computer infrastructure} asymmetry across regions. 

We conclude that cross-region DEX may be feasible if the asymmetry levels in computer infrastructure and trading interest across regions are correlated, which is a natural assumption since both are driven by economic activity, which tends to be clustered geographically. Our results are relevant for exchange operators, FinTech enterpreneurs, and financial regulators.


\ifCLASSOPTIONcompsoc
  \section*{Acknowledgments}
\else
  \section*{Acknowledgment}
\fi

Marius Zoican gratefully acknowledges the Connaught Foundation for a New Researcher Award.

\appendix

\begin{center}
\begin{tabular}{@{}ll@{}}
\toprule
\multicolumn{2}{c}{\textbf{Model parameters and interpretation}}\\
\cmidrule{1-2}
Parameter & Definition \\
\cmidrule{1-2}
$N$ & Number of miners in each cluster, \textbf{A} and \textbf{B}. \\
$\beta\geq \frac{1}{2}$ & Fraction of trades originating in cluster \textbf{A}.\\
$f$ & Exogenous mining fee. \\
$\pi\leq 1$ & Likelihood ratio: $\dfrac{\mathbb{P}\left(\text{long-distance miner wins }\right)}{\mathbb{P}\left(\text{local miner wins}\right)}$ \\
$\Delta$ & Efficiency gains from a single DEX. \\
$\theta$ & Inter-cluster network latency. \\
\bottomrule
\end{tabular}
\end{center}

\begin{center}
\begin{tabular}{@{}ll@{}}
\toprule
\multicolumn{2}{c}{\textbf{Simulation variables and interpretation.}}\\
\cmidrule{1-2}
Variable & Definition \\
\cmidrule{1-2}
$\eta_i$ & Number of nodes in cluster $i$, $i\in\left\{\textbf{A},\textbf{B}\right\}$  \\
$F^\text{full}_\ell \left(\tau\right)$ & Network-wide distribution of message delays \\ 
$F^{i\in\left\{\textbf{A},\textbf{B}\right\}}_\ell \left(\cdot\right)$ & Cluster-specific distribution of message delays \\ 
$\hat{p}^\text{cluster}_{i\in\left\{\textbf{A},\textbf{B}\right\}}$ & Probability that winning node is in cluster $i$. \\
$\hat{p}^\text{node}_{i\in\left\{\textbf{A},\textbf{B}\right\}}$ & Win probability for a given node in cluster $i$. \\

\bottomrule
\end{tabular}
\end{center}



%

\bibliographystyle{IEEEtran}
\bibliography{ieee}

\end{document}